\documentclass[11pt]{article}
\usepackage[margin=1in]{geometry}
\usepackage[T1]{fontenc}
\usepackage{lmodern,microtype}
\usepackage{amsmath,amssymb,amsthm}
\usepackage{xcolor}
\usepackage{float}
\usepackage{needspace}
\usepackage{tikz}
\usetikzlibrary{arrows.meta}
\usepackage[hidelinks]{hyperref}
\usepackage[nameinlink,noabbrev]{cleveref}

\newtheorem{theorem}{Theorem}[section]
\newtheorem{lemma}[theorem]{Lemma}
\newtheorem{corollary}[theorem]{Corollary}
\newtheorem{claim}[theorem]{Claim}
\crefname{claim}{claim}{claims}
\Crefname{claim}{Claim}{Claims}
\newcommand{\Z}{\mathbb Z}
\newcommand{\OO}{\widetilde O}
\newcommand{\ind}{\mathbf 1}

\title{A simplified proof of polynomial improvements
for 3SUM and APSP}
\author{
Sankeerth Rao Karingula\\
{Agentin AI}\\
{\texttt{sankeerthrao@agentin.ai}}
\and
Shachar Lovett\thanks{Supported by Simons Investigator Award \#929894 and NSF award CCF-2425349.}\\
{University of California, San Diego}\\
{\texttt{slovett@ucsd.edu}}
}

\begin{document}
\maketitle
\begin{abstract}
We give a simplified proof of the polynomial improvements for integer
3SUM and all-pairs shortest paths (APSP) established by Alman and
Vassilevska Williams, with weaker exponents. The central ingredient is
selected-entry matrix multiplication: computing a prescribed set of
entries of a product of dense rectangular matrices. Starting from
Sch\"onhage's inner/outer-product identity, we construct a multiplication
rule that shares intermediate products among many outputs. We apply
this rule recursively and bound the work by counting the requested
entries at each level.

Applied to incidence matrices, the algorithm solves Offline Set
Disjointness over a small universe. Known deterministic reductions
then yield truly subquadratic integer 3SUM and truly subcubic APSP
for polynomially bounded integer inputs.
\end{abstract}

\section{Introduction}\label{sec:introduction}

Alman and Vassilevska Williams~\cite{alman2026truly} established the first
truly subquadratic algorithms for integer 3SUM and truly subcubic
algorithms for all-pairs shortest paths with polynomially bounded
integer inputs. Their main innovation is a matrix multiplication
algorithm that exploits sparsity in the set of requested outputs,
even when the input matrices are dense. We give a simplified proof
of these qualitative improvements, with weaker exponents.

We focus on the matrix multiplication argument. To obtain the 3SUM
and APSP consequences, we first apply it to Offline Set Disjointness,
where matrix entries directly encode intersection sizes. We then use
existing deterministic reductions from 3SUM and APSP to this problem
as black boxes.

\paragraph{Selected-entry matrix multiplication.}
Let $X$ be an $n\times d$ matrix and $Y$ a $d\times m$ matrix.
Each entry of $XY$ is a dot product of a row of $X$ and a column
of $Y$. We call the task of computing $w$ specified entries of $XY$
\emph{selected-entry matrix multiplication}. Computing these entries
as separate dot products uses $wd$ multiplications, while writing
down the full product takes $\Omega(nm)$ time. When the middle
dimension $d$ and the number of requested entries $w$ are small enough,
sharing intermediate computations gives the following bound.

\begin{theorem}[Selected-entry matrix multiplication]\label{thm:matrix_mult}
For every fixed $\tau\in(0,1]$, there is an integer
$\kappa=\kappa(\tau)=O(\tau^{-2})$ with the following property.
Let $d^\kappa\le\min\{n,m\}$. Given integer matrices
$X\in\Z^{n\times d}$ and $Y\in\Z^{d\times m}$, with entries of
polynomial magnitude in $n+m$, and $w\ge1$ distinct specified positions
in $XY$, those entries can be computed deterministically in time
\begin{equation}\label{eq:partial}
 \OO\bigl((n+m)d^\kappa+d^\tau\sqrt{nmw}\bigr).
\end{equation}
\end{theorem}

The first term pays for preparing linear combinations of the inputs;
the second pays for multiplying these forms and reconstructing the
requested outputs.
Reducing $\tau$ makes the latter cheaper, at the expense of increasing
$\kappa$ and restricting the allowable middle dimension $d$.
Alman and Vassilevska Williams~\cite{alman2026truly} call such products
\emph{thin} and the prescribed outputs \emph{wanted entries}.

For example, take $\tau=1/4$, let $\kappa=\kappa(1/4)$, and suppose
$d^{\kappa+1}\le n$.
Computing $w=\lfloor n^2/d\rfloor$ specified entries of an $n\times d$ by $d\times n$
product then takes $\OO(n^2/d^{1/4})$ time. Separate dot products
would take $\Theta(n^2)$ operations, and the full product has $n^2$
entries. When $d$ is a fixed positive power of $n$, this saves a
polynomial factor over both approaches.

\paragraph{Fine grained algorithms.}
Integer 3SUM asks whether a list of $n$ integers contains three entries
with distinct indices whose sum is zero.
All-pairs shortest paths (APSP) asks for the shortest-path distance
between every ordered pair of vertices in a directed graph.
We allow negative edge weights but assume there is no negative cycle.
Combining \Cref{thm:matrix_mult} with known reductions gives the following
corollary.

\begin{corollary}\label{cor:apsp_3sum}
There are absolute constants $\delta_1,\delta_2>0$ such that integer
3SUM on $n$ polynomially bounded integers has a deterministic
$O(n^{2-\delta_1})$-time algorithm, and APSP on an $n$-vertex directed
graph with polynomially bounded integer edge weights and no negative
cycle has a deterministic $O(n^{3-\delta_2})$-time algorithm.
\end{corollary}

We have not tried to optimize the constants; our focus is on simplicity of exposition.

\subsection{Proof overview}

\paragraph{Selected-entry matrix multiplication.}
The algebraic starting point is the inner/outer-product identity of
Sch\"onhage~\cite{schonhage1981partial}, used by Coppersmith~\cite{coppersmith1982rapid} for
rectangular matrix multiplication. Alman and Vassilevska
Williams~\cite{alman2026truly} adapt this approach to exploit the
overlap among the computations needed by selected entries.
We express this sharing through an exact multiplication rule:
each output is a sum of six products from a common collection,
minus one private correction term. The common collection is small
relative to the number of outputs, and the same intermediate value
can serve many requested entries. We apply the rule recursively to
matrix blocks and count the resulting requests, capping each count
by the number of available output positions. Our simplification is
in the explicit rule and this counting argument; the central
idea of exploiting sparse output requests comes from
Alman and Vassilevska Williams~\cite{alman2026truly}.

\paragraph{Fine grained algorithms.}
In \emph{Offline Set Disjointness}~\cite{kopelowitz2016higher,fischer2023deterministic},
we are given two families of sets and a list of pairs, one set from
each family, and must determine which pairs are disjoint. The product
of the families' incidence matrices records all intersection sizes.
When the universe is small, \Cref{thm:matrix_mult} computes the requested
intersection sizes efficiently.
For 3SUM, we use the direct deterministic reduction of
Fischer, Kaliciak, and Polak~\cite{fischer2023deterministic}.
For APSP, we use the standard reduction to Exact Triangle, followed by
the deterministic reduction to lopsided triangle queries of
Alman and Vassilevska Williams~\cite{alman2026truly}.

\paragraph{Further research.}
How far can selected-entry matrix multiplication be improved?
What is the best possible dependence on the middle dimension $d$ when
computing $w$ arbitrary entries of an $n\times d$ by $d\times n$ product?
Can we reduce the $d^\tau$ loss in
\Cref{thm:matrix_mult} without increasing the input-preparation cost
or requiring thinner matrices? Can comparable savings hold for larger
$d$? For which ranges of $d$ and $w$ can selected-entry algorithms
substantially outperform the best rectangular matrix multiplication
algorithms for computing the full product?
It would also be valuable to find simpler multiplication rules or
entirely different algorithms achieving such savings. Progress on these
questions may also lead to stronger algorithms for 3SUM, APSP, and other fine grained algorithmic problems.

\paragraph{AI methodology.}
Agentin AI coordinated a large research swarm of AI agents to explore alternative proof strategies and candidate simplifications,
contributing to the development of the simplified proof presented
here. The authors directed the mathematical investigation,
critically evaluated the proposed arguments, and developed and
refined the proof and exposition through iterative work with
ChatGPT and Codex. This process included checking calculations,
simplifying the recursive analysis, separating input-preparation
costs from multiplication costs, and improving the presentation.
The authors take responsibility for the final manuscript.

\paragraph{Acknowledgements.} We thank Virginia Vassilevska Williams for helpful comments.

\paragraph{Organization.}
After brief preliminaries, \Cref{sec:sharing} proves the selected-entry
bound, and \Cref{sec:applications} applies it to set disjointness and derives
the 3SUM and APSP consequences using known reductions.

\paragraph{Preliminaries.}

Throughout, $[t]=\{1,\ldots,t\}$ and logarithms have base two.
We write $\OO(f)$ for $O(f)$ up to polylogarithmic factors.
All integer inputs are polynomially bounded in the relevant size
parameter, so they have logarithmic bit length.
Arithmetic on such integers costs only polylogarithmic time, which
is absorbed by $\OO$.

\section{Selected-entry matrix multiplication}\label{sec:sharing}

\Cref{thm:matrix_mult} follows from the main theorem, which gives a more precise bound on the runtime of the selected-entry matrix multiplication algorithm.

\begin{theorem}[Combinatorial bound]\label{thm:partial-count}
Fix an integer $k\ge81$. Let $d^{6k}\le\min\{n,m\}$ and
$t=\lceil\log_2 d\rceil$. Given integer matrices
$X\in\Z^{n\times d}$ and $Y\in\Z^{d\times m}$ with entries of
polynomial magnitude in $n+m$, any $w\ge1$ distinct specified entries
of $XY$ can be computed deterministically in time
\begin{equation}\label{eq:partial-count}
 \OO\left((n+m)d^{6k}
   +\sum_{s=0}^t\binom ts6^s\min\left\{w,\frac{nm}{k^s}\right\}\right).
\end{equation}
\end{theorem}

The following convenient corollary simplifies the runtime in terms of a single parameter $\theta \in [0,1]$.

\begin{corollary}[Density tradeoff]\label{cor:partial-tradeoff}
Under the assumptions of \Cref{thm:partial-count}, the same algorithm
runs in time
\begin{equation}\label{eq:partial-tradeoff}
 \OO\left((n+m)d^{6k}
   +(nm)^\theta w^{1-\theta}d^{\log_2(1+6/k^\theta)}\right)
 \qquad(0\le\theta\le1).
\end{equation}
\end{corollary}
\begin{proof}
Using $\min\{x,y\}\le x^{1-\theta}y^\theta$ and the binomial theorem,
\[
 \sum_{s=0}^t\binom ts6^s\min\left\{w,\frac{nm}{k^s}\right\}
 \le(nm)^\theta w^{1-\theta}
       \sum_{s=0}^t\binom ts(6/k^\theta)^s
 =(nm)^\theta w^{1-\theta}(1+6/k^\theta)^t.
\]
Since $t=\lceil\log_2 d\rceil$, the last factor is
$O(d^{\log_2(1+6/k^\theta)})$.
\end{proof}

\begin{proof}[Proof of \Cref{thm:matrix_mult}]
Take $k=\lceil81/\tau^2\rceil$ and $\kappa=6k$.
Then $1+6/\sqrt k\le1+2\tau/3<2^\tau$, so
\Cref{cor:partial-tradeoff} with $\theta=1/2$ gives the claim.
\end{proof}

We prove \Cref{thm:partial-count} in the remainder of this section.

\paragraph{Proof idea.}
We first construct a small matrix multiplication rule in which outputs
use six products from the same shared collection minus one private product per output. We apply this rule
recursively to blocks of larger matrices, following only the intermediate
entries needed for the requested outputs. We count these requests by
their pattern of private and shared steps. Each shared step generates
at most six requests per incoming request, but the count is capped by
the total number of output positions in the corresponding products.
The base case is an outer product, whose requested entries each
require one multiplication. Padding the input dimensions lets us
apply this count directly to the given matrices.

\subsection{The multiplication rule}\label{sec:rule-construction}

For a matrix product $UV$, call $U$ the left input and $V$ the
right input. A \emph{product of linear forms} multiplies an integer
linear combination of entries of $U$ by an integer linear combination
of entries of $V$. We count the additions and subtractions needed
to form these combinations separately from their multiplication.

\begin{lemma}[A shared multiplication rule]\label{lem:rule}
Let $k \ge 2$ and set $R = k 2^{k-1}$ and $S=3^{k-1}$. We construct a multiplication rule for matrices $U\in\Z^{R\times2}$ and
 $V\in\Z^{2\times S}$ as follows:
\begin{enumerate}
 \item \textbf{Products:} There is a common collection of $6^k$ \emph{shared}
 products and one \emph{private} product for each entry of $UV$.
 \item \textbf{Linear forms:} Every product multiplies a linear form
 in $U$ by a linear form in $V$. Each form is a sum of at most $6k$
 signed entries of its respective matrix.
 \item \textbf{Reconstruction:} Each entry of $UV$ is the sum of six
 shared products minus its private product.
\end{enumerate}
\end{lemma}

We count products by their labels, even if two formulas happen to be identical.
The multiplication has $RS=k6^{k-1}$ outputs, but only $6^k$ shared
products.  Thus
\[
 \frac{\text{number of shared products}}{\text{number of outputs}}
 =\frac{6^k}{RS}=\frac6k.
\]
As $k$ increases, the shared collection becomes small relative to the
number of outputs. This is the reuse that our recursive
algorithm will exploit.

\par\medskip\noindent\textit{A cancellation identity.}\quad
The construction follows the Sch\"onhage-Coppersmith approach
used by Alman and Vassilevska Williams~\cite{alman2026truly}.
We first give a cancellation identity, then construct shared linear
forms whose restriction to each output has exactly this pattern.
The identity is the
$2\times3$ instance of Sch\"onhage's inner/outer-product
cancellation~\cite{schonhage1981partial}.

For the desired dot product $p_1q_1+p_2q_2$, we add a third coordinate with
$p_3=0$ and $q_3=-q_1-q_2$. The first relation preserves the dot
product, while the zero sum of the $q_j$ cancels a mixed term.
The remaining unwanted terms are removed by one private product.
\Cref{fig:cancellation} illustrates the identity.

\begin{claim}[Cancellation]\label{clm:cancellation}
Let $x_1,x_2,p_1,p_2$ be left variables and
$y_1,y_2,y_3,q_1,q_2$ be right variables. Put $p_3=0$ and
$q_3=-q_1-q_2$. Then
\begin{equation}\label{eq:cancellation}
 \sum_{j=1}^3\bigl[(x_1+p_j)(y_j+q_j)+(x_2-p_j)y_j\bigr]-(x_1+x_2)(y_1+y_2+y_3)=p_1q_1+p_2q_2.
\end{equation}
\end{claim}
\begin{proof}
For each $j$, the expression in brackets is
$(x_1+x_2)y_j+x_1q_j+p_jq_j$.
Summing over $j$, the first terms cancel the subtracted product,
the second sum to zero because $q_1+q_2+q_3=0$, and the last
sum to $p_1q_1+p_2q_2$ because $p_3=0$.
\end{proof}

\begin{figure}[htbp]
 \centering
\begingroup
\definecolor{cancelprivate}{HTML}{176AAB}
\definecolor{cancelzero}{HTML}{B8660D}
\definecolor{canceltarget}{HTML}{237A57}
\definecolor{cancelgray}{HTML}{66717C}
\definecolor{cancelgrid}{HTML}{DCE2E8}
\begin{tikzpicture}[
 x=1cm,y=1cm,
 every node/.style={inner sep=1pt},
 note/.style={font=\small,align=center}
]
 \path[use as bounding box] (-.05,.28) rectangle (15.55,-4.95);

 \fill[cancelprivate!5] (1.25,-.35) rectangle (14.75,-1.25);
 \draw[cancelgrid,line width=.6pt]
   (1.25,-.35) rectangle (14.75,-2.15);
 \draw[cancelgrid,line width=.6pt] (1.25,-1.25)--(14.75,-1.25);
 \foreach \x in {5.75,10.25}
   \draw[cancelgrid,line width=.6pt] (\x,-.35)--(\x,-2.15);
 \foreach \j [evaluate=\j as \x using 3.5+4.5*(\j-1)] in {1,2,3} {
   \node[text=cancelgray] at (\x,0) {$j=\j$};
   \node at (\x,-.80) {$(x_1+p_{\j})(y_{\j}+q_{\j})$};
   \node at (\x,-1.70) {$(x_2-p_{\j})y_{\j}$};
 }
 \node[text=cancelgray] at (.50,-.80) {$i=1$};
 \node[text=cancelgray] at (.50,-1.70) {$i=2$};

 \draw[-{Stealth[length=5pt]},cancelgray,line width=.7pt]
   (8,-2.36)--(8,-2.96);
 \node[anchor=west,font=\small,text=cancelgray] at (8.22,-2.66)
   {sum all six};

 \node[text=cancelprivate] at (2.65,-3.42)
   {$(x_1+x_2)(y_1+y_2+y_3)$};
 \node at (5.35,-3.42) {$+$};
 \node[text=cancelzero] at (8.00,-3.42)
   {$x_1(q_1+q_2+q_3)$};
 \node at (10.67,-3.42) {$+$};
 \node[text=canceltarget] at (13.15,-3.42)
   {$p_1q_1+p_2q_2$};

 \draw[cancelprivate,line width=1pt] (.55,-3.82)--(4.75,-3.82);
 \draw[cancelzero,line width=1pt] (6.15,-3.82)--(9.85,-3.82);
 \draw[canceltarget,line width=1pt] (11.45,-3.82)--(14.85,-3.82);
 \node[note,text=cancelprivate] at (2.65,-4.40)
   {Removed by subtracting\\one private product};
 \node[note,text=cancelzero] at (8.00,-4.40)
   {Zero by the\\zero-sum relation};
 \node[note,text=canceltarget] at (13.15,-4.40)
   {The desired\\dot product};
\end{tikzpicture}
\endgroup
 \caption{The six products in \cref{eq:cancellation} and their sum,
 with $p_3=0$ and $q_1+q_2+q_3=0$. Subtracting one private product
 leaves the desired dot product.}
 \label{fig:cancellation}
\end{figure}

Note that the variables $x_i,y_j$ can be arbitrary: the identity always recovers
the same two-term dot product. We will choose them so that the six
positive products for different outputs belong to one small collection.

\par\medskip\noindent\textit{The shared products.}\quad
Fix the integer $k\ge2$ from \Cref{lem:rule} for the rest of
this construction. Recall that $R=k2^{k-1}$ and $S=3^{k-1}$, and
let $U\in\Z^{R\times2}$ and $V\in\Z^{2\times S}$.
Label the rows of $U$ by
$(h,u)$, with $h\in[k]$ and $u\in\{1,2\}^{k-1}$, and the
columns of $V$ by $v\in\{1,2,3\}^{k-1}$. To match the three
values of $j$ in \Cref{clm:cancellation}, extend the entry
notation by
\[
 U_{(h,u),3}=0,
 \qquad V_{3,v}=-V_{1,v}-V_{2,v}.
\]
These are abbreviations for linear forms in the original entries.

For a tuple $a$ of length $k$, write $a_{-\ell}$ for the tuple
obtained by deleting coordinate $\ell$, keeping the others in their
original order. For example, $(a_1,a_2,a_3)_{-2}=(a_1,a_3)$.

For each $a\in\{1,2\}^k$ and $b\in\{1,2,3\}^k$, define
\begin{equation}\label{eq:shared-forms}
 A_{a,b}=\sum_{\ell=1}^k(-1)^{a_\ell-1}
                 U_{(\ell,a_{-\ell}),b_\ell},
 \qquad
 B_{a,b}=\sum_{\ell=1}^k\ind_{a_\ell=1}
                 V_{b_\ell,b_{-\ell}}.
\end{equation}
Here $\ind_{a_\ell=1}$ is $1$ if $a_\ell=1$ and $0$ otherwise.
Our shared collection consists of the $6^k$ products $A_{a,b}B_{a,b}$.
Each multiplies a linear form in $U$ by a linear form in $V$.

The output $((h,u),v)$ uses the six products whose labels satisfy
$a_{-h}=u$ and $b_{-h}=v$. Only the pair $(a_h,b_h)$ varies, so
these labels form a six-point coordinate line. Each label belongs
to one such line for each $h\in[k]$; thus each shared product serves
$k$ outputs. \Cref{fig:shared-products} illustrates this pattern for $k=2$.

Along this line, the $\ell=h$ terms in \cref{eq:shared-forms} supply
the entries of the desired dot product. The remaining terms play
the roles of $x_i,y_j$ in \Cref{clm:cancellation}, so one private
product cancels their contribution. We verify this next.

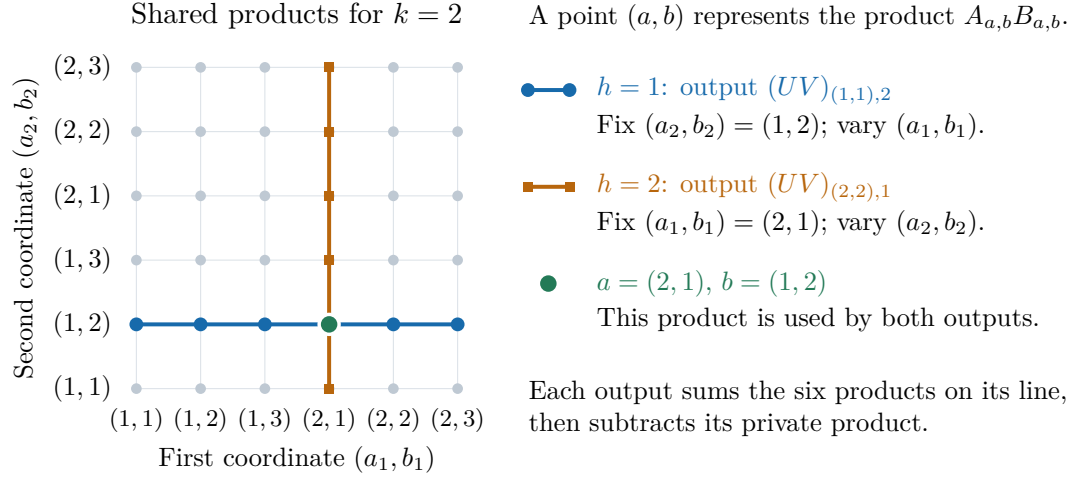
\begin{figure}[htbp]
 \centering
\begingroup
\definecolor{sharedrow}{HTML}{176AAB}
\definecolor{sharedcolumn}{HTML}{B8660D}
\definecolor{sharedpoint}{HTML}{237A57}
\definecolor{sharedgray}{HTML}{66717C}
\definecolor{sharedgrid}{HTML}{DCE2E8}
\definecolor{shareddot}{HTML}{BFC9D3}
\begin{tikzpicture}[
 x=1cm,y=1cm,
 every node/.style={inner sep=1pt},
 detail/.style={font=\small,align=left,anchor=west}
]
 \path[use as bounding box] (-.20,-1.27) rectangle (15.80,5.90);
 \begin{scope}[xshift=2mm]
 \node at (3.325,5.53) {Shared products for $k=2$};

 \foreach \n in {0,...,5} {
   \pgfmathsetmacro{\coord}{1.20+.85*\n}
   \draw[sharedgrid,line width=.4pt] (\coord,.60)--(\coord,4.85);
   \pgfmathsetmacro{\coord}{.60+.85*\n}
   \draw[sharedgrid,line width=.4pt] (1.20,\coord)--(5.45,\coord);
 }
 \foreach \i in {0,...,5}
   \foreach \j in {0,...,5}
     \fill[shareddot] (1.20+.85*\i,.60+.85*\j) circle[radius=2pt];
 \foreach \a/\b [count=\n from 0] in {1/1,1/2,1/3,2/1,2/2,2/3} {
   \node[font=\footnotesize] at (1.20+.85*\n,.18) {$(\a,\b)$};
   \node[font=\small,anchor=east] at (.92,.60+.85*\n) {$(\a,\b)$};
 }
 \node[font=\small] at (3.325,-.33) {First coordinate $(a_1,b_1)$};
 \node[font=\small,rotate=90] at (-.26,2.725) {Second coordinate $(a_2,b_2)$};

 \draw[sharedrow,line width=1.5pt] (1.20,1.45)--(5.45,1.45);
 \foreach \i in {0,...,5}
   \fill[sharedrow] (1.20+.85*\i,1.45) circle[radius=2.6pt];

 \draw[sharedcolumn,line width=1.5pt] (3.75,.60)--(3.75,4.85);
 \foreach \j in {0,...,5}
   \fill[sharedcolumn]
     (3.685,.535+.85*\j) rectangle (3.815,.665+.85*\j);

 \fill[white] (3.75,1.45) circle[radius=4.4pt];
 \fill[sharedpoint] (3.75,1.45) circle[radius=3.1pt];
 \end{scope}

 \node[detail] at (6.55,5.50)
   {A point $(a,b)$ represents the product $A_{a,b}B_{a,b}$.};

 \draw[sharedrow,line width=1.5pt] (6.58,4.55)--(7.13,4.55);
 \fill[sharedrow] (6.58,4.55) circle[radius=2.6pt];
 \fill[sharedrow] (7.13,4.55) circle[radius=2.6pt];
 \node[detail,text=sharedrow] at (7.45,4.55)
   {$h=1$: output $(UV)_{(1,1),2}$};
 \node[detail] at (7.45,4.07)
   {Fix $(a_2,b_2)=(1,2)$; vary $(a_1,b_1)$.};

 \draw[sharedcolumn,line width=1.5pt] (6.58,3.27)--(7.13,3.27);
 \fill[sharedcolumn] (6.515,3.205) rectangle (6.645,3.335);
 \fill[sharedcolumn] (7.065,3.205) rectangle (7.195,3.335);
 \node[detail,text=sharedcolumn] at (7.45,3.27)
   {$h=2$: output $(UV)_{(2,2),1}$};
 \node[detail] at (7.45,2.79)
   {Fix $(a_1,b_1)=(2,1)$; vary $(a_2,b_2)$.};

 \fill[sharedpoint] (6.85,1.99) circle[radius=3.1pt];
 \node[detail,text=sharedpoint] at (7.45,1.99)
   {$a=(2,1)$, $b=(1,2)$};
 \node[detail] at (7.45,1.51)
   {This product is used by both outputs.};

 \node[detail] at (6.55,.55)
   {Each output sums the six products on its line,};
 \node[detail] at (6.55,.12)
   {then subtracts its private product.};
 \node[font=\small] at (7.85,-.94)
   {In $k$ dimensions: six points per line, and $k$ lines through each point.};
\end{tikzpicture}
\endgroup
 \caption{Shared products and outputs for $k=2$. Each point is one
 shared product, and each coordinate line selects the six shared
 products for one output. The highlighted lines reuse the product
 at their intersection.}
 \label{fig:shared-products}
\end{figure}

\begin{proof}[Proof of \Cref{lem:rule}]
Fix an output $((h,u),v)$. Let $a(i)$ be obtained by inserting
$i\in\{1,2\}$ into $u$ at position $h$, and let $b(j)$ be
obtained by inserting $j\in\{1,2,3\}$ into $v$ at position $h$.
For example, for $k=3$ and $h=2$,
\[
 a(i)=(u_1,i,u_2),\qquad b(j)=(v_1,j,v_2).
\]
In particular, $a(i)_{-h}=u$ and $b(j)_{-h}=v$.
Use the six shared products with labels $(a(i),b(j))$, and put
\[
 p_j=U_{(h,u),j},\qquad q_j=V_{j,v}\qquad(1\le j\le3).
\]
Our extended notation gives $p_3=0$ and $q_3=-q_1-q_2$, as required.
In~\cref{eq:shared-forms}, the $\ell=h$ terms contribute
$(-1)^{i-1}p_j$ to $A_{a(i),b(j)}$ and $\ind_{i=1}q_j$ to
$B_{a(i),b(j)}$. For $\ell\ne h$, the terms in $A$ are
independent of $j$, since only $b_h$ varies with $j$; call their
sum $x_i$. Likewise, the remaining terms in $B$ are independent
of $i$, since only $a_h$ varies with $i$; call their sum $y_j$.
Thus
\[
\begin{aligned}
 A_{a(1),b(j)}&=x_1+p_j,&\qquad A_{a(2),b(j)}&=x_2-p_j,\\
 B_{a(1),b(j)}&=y_j+q_j,& B_{a(2),b(j)}&=y_j.
\end{aligned}
\]
Applying \Cref{clm:cancellation} gives
\[
 (UV)_{(h,u),v}
 =\sum_{i=1}^2\sum_{j=1}^3 A_{a(i),b(j)}B_{a(i),b(j)}
   -(x_1+x_2)(y_1+y_2+y_3).
\]
The subtracted term is one product of a linear form in $U$ and a
linear form in $V$; assign it as the private product for this output.
All six positive terms belong to the fixed collection of $6^k$
shared products.

For the bound on the linear forms, expand each abbreviation
$V_{3,v}=-V_{1,v}-V_{2,v}$. Each shared form and each $x_i,y_j$
is then a sum of at most $2k$ signed input entries. Each private
factor combines at most three such forms, giving the bound $6k$.
\end{proof}

\subsection{The recursive computation}

We now use \Cref{lem:rule} to reduce the middle dimension of
a matrix product by a factor of two. A \emph{subproblem} consists
of two input matrices and a set of entries of their product that
we want to compute. The original subproblem is to multiply
$X\in\Z^{n\times d}$ by $Y\in\Z^{d\times m}$ at the $w$
specified positions.

We first explain one reduction step, then bound three costs after
iterating it: preparing input linear forms, multiplying them, and
reconstructing the requested outputs. Padding the input dimensions
will then give \Cref{thm:partial-count}. We first count
arithmetic operations on matrix values, and account for request
processing and intermediate magnitudes at the end.

\par\medskip\noindent\textit{One reduction step.}\quad
Fix an integer $k\ge2$ and let $R=k2^{k-1}$ and $S=3^{k-1}$,
as in \Cref{lem:rule}.

\begin{claim}[One reduction step]\label{clm:reduction-step}
Consider matrices $X\in\Z^{N\times D}$ and $Y\in\Z^{D\times M}$,
with $w$ distinct specified entries of $XY$. Suppose $R\mid N$, $S\mid M$,
and $2\mid D$. This subproblem reduces to $RS=k6^{k-1}$ private subproblems
and $6^k$ shared subproblems. Each multiplies matrices of dimensions
\[
 \left(\frac NR\times\frac D2\right)
 \quad\text{and}\quad
 \left(\frac D2\times\frac MS\right).
\]
The private subproblems receive $w$ requested entries in total,
and the shared subproblems receive at most $6w$. Their answers determine
the $w$ original entries using at most $6w$ additions and subtractions.
\end{claim}

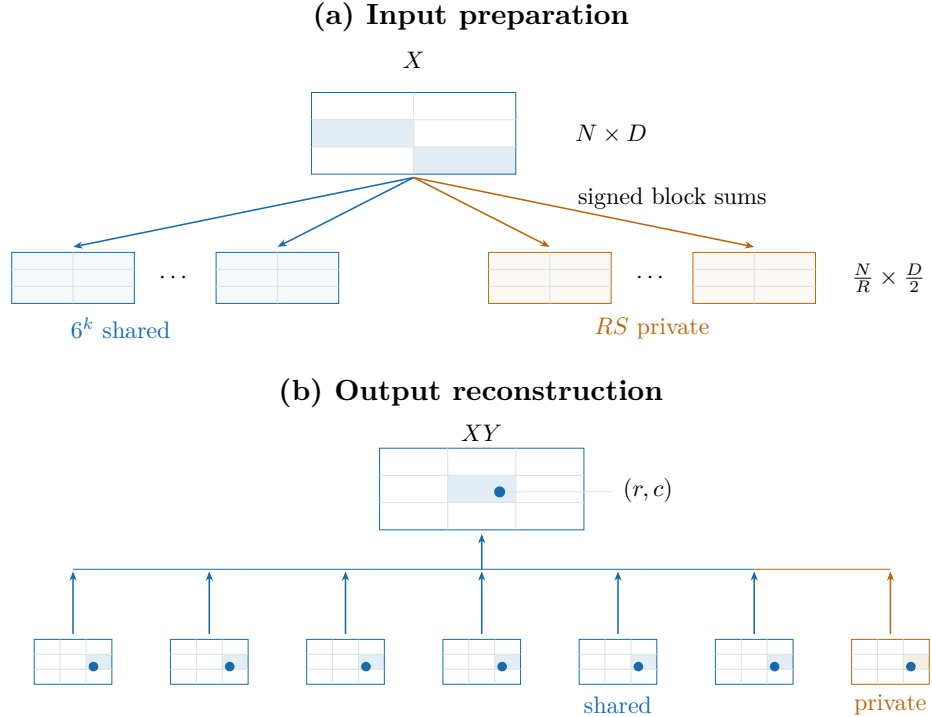
\begin{figure}[H]
 \centering
 \textbf{(a) Input preparation}\par\smallskip
 \scalebox{0.9}{
\begingroup
\definecolor{flowblue}{HTML}{176AAB}
\definecolor{flowprivate}{HTML}{B8660D}
\definecolor{flowgrid}{HTML}{DCE2E8}
\begin{tikzpicture}[x=1cm,y=1cm,
 every node/.style={inner sep=2pt,font=\small},
 arrow/.style={-{Stealth[length=4pt]},flowblue,line width=.65pt}]
 \path[use as bounding box] (0,1.90) rectangle (16,6.4);
 \node at (7,6.13) {$X$};
 \fill[flowblue!12] (5.5,4.85) rectangle (7,5.25);
 \fill[flowblue!12] (7,4.45) rectangle (8.5,4.85);
 \draw[flowblue] (5.5,4.45) rectangle (8.5,5.65);
 \draw[flowgrid] (7,4.45)--(7,5.65);
 \foreach \y in {4.85,5.25} \draw[flowgrid] (5.5,\y)--(8.5,\y);
 \node[anchor=west] at (9.3,5.05) {$N\times D$};
 \node[fill=white] at (10.8,4.10) {signed block sums};
 \foreach \x/\col in {2/flowblue,5/flowblue,9/flowprivate,12/flowprivate} {
  \draw[arrow,\col] (7,4.40)--(\x,3.38);
  \draw[\col,fill=\col!4] (\x-.9,2.55) rectangle (\x+.9,3.30);
  \draw[flowgrid] (\x,2.55)--(\x,3.30);
  \foreach \y in {2.80,3.05} \draw[flowgrid] (\x-.9,\y)--(\x+.9,\y);
 }
 \node at (3.5,2.925) {$\cdots$};
 \node at (10.5,2.925) {$\cdots$};
 \node[text=flowblue] at (2.7,2.19) {$6^k$ shared};
 \node[text=flowprivate] at (10.5,2.19) {$RS$ private};
 \node[anchor=west] at (13.35,2.925) {$\frac NR\times\frac D2$};
\end{tikzpicture}
\endgroup}
 \par\medskip
 \textbf{(b) Output reconstruction}\par\smallskip
 \scalebox{0.9}{
\begingroup
\definecolor{answerblue}{HTML}{176AAB}
\definecolor{answerprivate}{HTML}{B8660D}
\definecolor{answergrid}{HTML}{DCE2E8}
\begin{tikzpicture}[x=1cm,y=1cm,
 every node/.style={inner sep=2pt,font=\small},
 up/.style={-{Stealth[length=4pt]},line width=.65pt}]
 \path[use as bounding box] (0,1.80) rectangle (14.4,6.3);
 \node at (7.2,6.08) {$XY$};
 \fill[answerblue!12] (6.7,5.05) rectangle (7.7,5.45);
 \draw[answerblue] (5.7,4.65) rectangle (8.7,5.85);
 \foreach \x in {6.7,7.7} \draw[answergrid] (\x,4.65)--(\x,5.85);
 \foreach \y in {5.05,5.45} \draw[answergrid] (5.7,\y)--(8.7,\y);
 \fill[answerblue] (7.463,5.208) circle[radius=2.3pt];
 \node[anchor=west] at (9.2,5.23) {$(r,c)$};
 \draw[answergrid] (7.60,5.208)--(9.12,5.21);
 \draw[up,answerblue] (7.2,4.07)--(7.2,4.59);
 \draw[answerblue] (1.2,4.07)--(11.2,4.07);
 \draw[answerprivate] (11.2,4.07)--(13.2,4.07);
 \foreach \x [count=\i] in {1.2,3.2,5.2,7.2,9.2,11.2,13.2} {
  \ifnum\i=7 \def\acolor{answerprivate}
  \else \def\acolor{answerblue}\fi
  \draw[up,\acolor] (\x,3.12)--(\x,4.04);
  \fill[\acolor!12] (\x+.19,2.60) rectangle (\x+.57,2.82);
  \draw[\acolor] (\x-.57,2.38) rectangle (\x+.57,3.04);
  \foreach \dx in {-.19,.19}
    \draw[answergrid] (\x+\dx,2.38)--(\x+\dx,3.04);
  \foreach \y in {2.60,2.82}
    \draw[answergrid] (\x-.57,\y)--(\x+.57,\y);
  \fill[answerblue] (\x+.30,2.64) circle[radius=2pt];
 }
 \node[text=answerblue] at (9.2,2.08) {shared};
 \node[text=answerprivate] at (13.2,2.08) {private};
\end{tikzpicture}
\endgroup}
 \caption{One reduction step in \Cref{clm:reduction-step}.
 (a) Signed sums of at most $6k$ parent blocks form each child input,
 shown for $X$; $Y$ is treated in the same way.
 (b) One requested parent entry is recovered from six shared answers
 minus one private answer. Dots mark the same local position $(r,c)$
 in the parent block and child outputs.}
 \label{fig:reduction-step}
\end{figure}

\begin{proof}
Partition $X$ into blocks $X_{i,t}$, with $i\in[R]$ and
$t\in\{1,2\}$, each of size $(N/R)\times(D/2)$.
Partition $Y$ into blocks $Y_{t,j}$, with $j\in[S]$, each of
size $(D/2)\times(M/S)$. The $(i,j)$ block of the output is
\[
 (XY)_{i,j}=X_{i,1}Y_{1,j}+X_{i,2}Y_{2,j}.
\]
This is the same $R\times2$ by $2\times S$ multiplication as in
\Cref{lem:rule}, with matrices in place of scalars. The lemma
still applies: its formulas use only linear combinations and
products of a left input by a right input, always in that order.
It expresses each output block as six shared products minus one
private product. Each such product multiplies a linear combination
of the $X$ blocks by a linear combination of the $Y$ blocks, giving
the stated dimensions.

To request an entry of $XY$, locate its output block $(i,j)$ and
its position within that block. Request that same position in the
six shared products and the one private product used by this block.
The seven returned values recover the original entry with five
additions and one subtraction.

Each child maintains a set of requested positions, so repeated requests
are computed only once. Each output block has its own private product, so the $w$ distinct
original requests give exactly $w$ distinct private requests.
The original requests also generate $6w$ shared requests, so the shared
children have at most $6w$ distinct requests in total.
\end{proof}

\Needspace{6\baselineskip}
\par\medskip\noindent\textit{The bounds after $t$ steps.}\quad
For fixed $k$, iterating the rule $t$ times gives the basic dimensions
$R^t\times2^t$ and $2^t\times S^t$. We also allow multiples of
the outer dimensions, so that arbitrary inputs can be handled by padding.
We iterate the input preparation and output reconstruction steps shown
in \Cref{fig:reduction-step}.

\begin{lemma}[Iterated multiplication rule]\label{lem:iterated-rule}
Fix $k\ge81$, with $R,S$ as above.
For integers $t\ge0$ and $a,b\ge1$, write
\[
 N=aR^t,\qquad D=2^t,\qquad M=bS^t.
\]
For $0\le j\le t$, define the capped request count
\begin{equation}\label{eq:capped-request-count}
 B_j=\sum_{s=0}^j\binom js6^s\min\left\{w,\frac{NM}{k^s}\right\}.
\end{equation}
Given matrices $U\in\Z^{N\times D}$ and $V\in\Z^{D\times M}$,
any $w\ge1$ distinct specified entries of $UV$ can be computed with
the following operation bounds:
\begin{enumerate}
 \item \textbf{Input preparation:}
 $D^{3k}(N+M)$ additions and subtractions to prepare linear
 forms separately in $U$ and in $V$.
 \item \textbf{Scalar multiplication:}
 $B_t$ multiplications of a prepared form in $U$
 by one in $V$.
 \item \textbf{Output reconstruction:}
 $O(tB_t)$ additions and subtractions to recover
 the requested entries.
\end{enumerate}
The algorithm generates $O(tB_t)$ request records.
\end{lemma}

\begin{proof}
Apply \Cref{clm:reduction-step} recursively for $t$ steps.
Each step has $T=RS+6^k=(k+6)6^{k-1}$ child products, so the
full recursion has $T^t$ leaves, each an $a\times1$ by $1\times b$
outer product. We compute the $a$ left input values and $b$ right
input values at every leaf, whether or not that leaf receives requests.
We then multiply only the pairs needed for requested leaf entries
and reconstruct the requested outputs up the recursion tree.

\textbf{Input preparation.}
By \Cref{lem:rule}, each entry of a child's left input matrix is a
sum of at most $6k$ signed entries of its parent's left input matrix.
The same holds for the right input matrices.
Substituting these expressions through all $t$ levels writes each
left input value at a leaf as a sum of at most $(6k)^t$ signed entries
of the original matrix $U$, and each right input value as such a sum
of entries of $V$.
There are $a+b$ input values per leaf. Computing these sums directly
from the original inputs uses at most
\[
 T^t(a+b)(6k)^t
 \le D^{3k}(a+b)
 \le D^{3k}(N+M)
\]
additions and subtractions, since $6kT\le2^{3k}$ for $k\ge81$
and $D=2^t$.

\textbf{Requested computations.}
A request specifies one entry of the output matrix of one subproblem.
Starting with the $w$ specified entries of $UV$, we pass requests down
the recursion tree. Each request generates one request in a private
child and six requests in shared children, all at the same local
position. Requests for the same entry of the same child are merged.
Let $Q_j$ be the total number of distinct requests at depth $j$,
counting requests in different subproblems separately. We will show
that $Q_j\le B_j$.

Group the subproblems at depth $j$ by their pattern of shared and
private steps from the root. A pattern specifies the type of child
at each step, but not which child is chosen. Fix a pattern with $s$
shared steps and $j-s$ private steps. We bound its requests in two ways:
\begin{itemize}
 \item \emph{Requests generated.} Each original request has one
 continuation at a private step and at most six at a shared step.
 It therefore generates at most $6^s$ requests with this pattern,
 or at most $6^s w$ over all original requests.
 \item \emph{Available positions.} The full recursion tree has $RS$
 choices of private child and $6^k$ choices of shared child at each
 step. Thus $(RS)^{j-s}(6^k)^s$ subproblems have this pattern.
 Each has an output matrix of size $(N/R^j)\times(M/S^j)$, so their
 total number of output positions is
 \[
 (RS)^{j-s}(6^k)^s\frac{NM}{(RS)^j}
 =6^s\frac{NM}{k^s},
 \]
 using $RS=k6^{k-1}$. Merging repeated requests caps the number
 of distinct requests by this number of positions.
\end{itemize}
Taking the smaller bound for each pattern, and noting that there
are $\binom js$ patterns with $s$ shared steps, gives
\[
 Q_j\le\sum_{s=0}^j\binom js
       \min\left\{6^s w,\,6^s\frac{NM}{k^s}\right\}=B_j.
\]

At depth $t$, the middle dimension is one, so each requested entry
is computed by multiplying two prepared input values. This uses
$Q_t\le B_t$ scalar multiplications. We then work back up the tree:
each requested parent entry is the sum of its six shared answers
minus its private answer, using six additions or subtractions.
Request generation and reconstruction therefore each cost
$O(\sum_{j=0}^{t-1}Q_j)$. Finally, $B_j\le B_t$ for $j\le t$,
since each coefficient $\binom js$ in \cref{eq:capped-request-count}
is nondecreasing in $j$. These two costs are thus $O(tB_t)$.

\end{proof}

\begin{proof}[Proof of \Cref{thm:partial-count}]
If $d=1$, then $t=0$ and the sum in \cref{eq:partial-count} is $w$,
so compute the entries directly. Otherwise, set $D=2^t<2d$.
Pad the outer dimensions to the least multiples $N,M$ of $R^t,S^t$
that are at least $n,m$, respectively, and pad the middle dimension
to $D$. Since $R,S\le2^{3k}$ and $D<2d\le d^2$,
\[
 \max\{R^t,S^t\}\le D^{3k}\le d^{6k}\le\min\{n,m\}.
\]
Hence $N\le2n$ and $M\le2m$.
Apply \Cref{lem:iterated-rule}. Input preparation costs at most
$2(n+m)d^{6k}$. Since $NM\le4nm$, the capped count satisfies
\[
 B_t=\sum_{s=0}^t\binom ts6^s\min\left\{w,\frac{NM}{k^s}\right\}
 \le4\sum_{s=0}^t\binom ts6^s\min\left\{w,\frac{nm}{k^s}\right\}.
\]
The extra factor $O(t)$ for reconstruction and request generation,
and the cost of sorting and deduplicating requests, are absorbed
by $\OO$.

If the input entries have absolute value at most $L\ge1$, each
preparation or reconstruction step combines $O(k)$ values.
All intermediate values are therefore bounded by $(6k)^{O(t)}L^2$.
Since $k$ is fixed, $t=O(\log d)$, and $L$ is polynomially bounded,
these integers have $O(\log(n+m))$ bits. This proves
\cref{eq:partial-count}.
\end{proof}

\section{Applications to fine grained algorithms}\label{sec:applications}

Fine grained algorithms seek to improve the exponent in a problem's
running time. Here our goal is to go below quadratic time for integer
3SUM and below cubic time for APSP. We do this through a convenient
intermediate problem, Offline Set Disjointness, which asks whether
specified pairs of sets intersect.

This problem connects directly to our matrix multiplication algorithm:
the product of two incidence matrices records the intersection sizes,
and we only need the entries corresponding to the queried pairs.
Known reductions then turn a faster algorithm for this intermediate
problem into faster algorithms for 3SUM and APSP. We first give the
set-disjointness algorithm, then apply these reductions.

\subsection{Offline Set Disjointness}\label{sec:set-disjointness}

An instance of the \emph{Offline Set Disjointness} problem consists of two families
$S_1,\ldots,S_n,T_1,\ldots,T_n\subseteq[d]$ and a set
$W\subseteq[n]\times[n]$ of $w$ requested pairs. For each
$(i,j)\in W$, we must determine whether $S_i\cap T_j=\varnothing$.

\begin{lemma}\label{lem:set-disjointness}
There is an absolute integer $K\ge2$ such that, whenever $d^K\le n$
and $w\le n^2/\sqrt d$, Offline Set Disjointness can be solved
deterministically in time $\OO(n^2/d^{1/8})$.
\end{lemma}
\begin{proof}
Take $\tau=1/8$ in \Cref{thm:matrix_mult}, and set
$K=\kappa(1/8)+1$. Form the incidence matrices
$X\in\{0,1\}^{n\times d}$ and $Y\in\{0,1\}^{d\times n}$.
Since $(XY)_{ij}=|S_i\cap T_j|$, it suffices to compute the entries
indexed by $W$. The theorem gives running time
\[
 \OO\bigl(nd^{K-1}+nd^{1/8}\sqrt w\bigr)
 \le\OO\bigl(n^2/d+n^2/d^{1/8}\bigr)
 =\OO(n^2/d^{1/8}).
\]
Forming the incidence matrices and reading the queries costs
$O(nd+w)$, which is absorbed by this bound.
\end{proof}

\subsection{Known reductions}

To identify the offline set disjointness problem with its graph formulation, use one vertex
for each $S_i$, one for each $T_j$, and a middle vertex for each element
of $[d]$. Join sets to their elements and add an outer edge for every
requested pair. That edge lies in a triangle exactly when the two sets
are not disjoint. This is the \emph{Lopsided All-Edges Sparse Triangle}
formulation~\cite{williams2020monochromatic,alman2026truly}, with a small middle part.

For 3SUM, Fischer, Kaliciak, and Polak~\cite{fischer2023deterministic} derandomized the
set-disjointness reduction of Kopelowitz, Pettie, and Porat~\cite{kopelowitz2016higher}.
For Exact Triangle, Vassilevska Williams and Xu~\cite{williams2020monochromatic} gave a
randomized reduction to the lopsided problem. Chan and Xu~\cite{chan2024simpler}
gave a deterministic reduction to All-Edges Sparse Triangle, but its
construction does not ensure the small middle part needed here.
Alman and Vassilevska Williams~\cite{alman2026truly} obtain the deterministic
lopsided reduction by combining deterministic hashing with the earlier
graph constructions. We use these reductions as black boxes.

\begin{proof}[Proof of \Cref{cor:apsp_3sum}]
\emph{3SUM.}
Fischer, Kaliciak, and Polak~\cite{fischer2023deterministic} show that, for every fixed
$0<\eta<1$, a deterministic truly subquadratic algorithm for Offline
Set Disjointness with $O(n)$ sets, universe size $O(n^{2\eta})$,
set size $O(n^\eta)$, and $O(n^{2-\eta})$ queries gives a deterministic
truly subquadratic algorithm for 3SUM.
Choose $\eta>0$ with $2\eta K<1$.
After padding the families and universe to sizes $\Theta(n)$ and
$d=\Theta(n^{2\eta})$, respectively, the query bound is
$O(n^2/\sqrt d)$.
A constant number of query batches allows \Cref{lem:set-disjointness}
to solve these instances in $\OO(n^{2-\eta/4})$ time, giving the
3SUM conclusion.

\emph{APSP.}
Exact Triangle asks whether an integer-weighted tripartite graph with
$n$ vertices per part contains a triangle whose edge weights sum to
zero. We restrict to weights of magnitude $n^{O(1)}$. Alman and Vassilevska
Williams~\cite{alman2026truly} reduce it deterministically, for
$16\le d\le n$ and an integer $1\le g\le\sqrt d$, to $O(ng)$
set-disjointness instances with parameters $(n,d)$ and at most
$n^2/\sqrt d$ queries each, with additional time
\[
 \OO\bigl(n^3/g+n^{\omega+o(1)}d^{3/2}+n^2dg\bigr),
\]
where $\omega<3$ is the square matrix multiplication exponent.

Choose $d=\lfloor n^\varepsilon\rfloor$ for a sufficiently small fixed
$\varepsilon>0$, and set $g=\lceil d^{1/16}\rceil$.
By \Cref{lem:set-disjointness}, the oracle calls cost
$\OO(ng\cdot n^2/d^{1/8})=\OO(n^3/d^{1/16})$,
as does the first overhead term. Since $\omega<3$, choosing
$\varepsilon$ small enough also ensures $d^K\le n$ and makes the
other two overhead terms smaller than this bound. Thus Exact Triangle takes
$\OO(n^{3-\rho})$ time for some constant $\rho>0$.

The deterministic reduction of
Vassilevska Williams and Williams~\cite{vassilevska2009finding,williams2018subcubic} solves APSP in
$\OO(n^{3-\rho/3})$ time, using only the decision version of
Exact Triangle. Absorbing logarithmic factors gives the claimed
polynomial improvement.
\end{proof}

\bibliographystyle{unsrt}
\bibliography{3sum}
\end{document}